\documentclass[12pt,psfig]{iopart}
\usepackage{iopams}
\usepackage{epsf}
\usepackage{euscript}
\usepackage{graphicx}
\usepackage{setstack} 
\usepackage{amsthm}   
  \theoremstyle{plain}
  \newtheorem{thm}{Theorem}[section]

  \theoremstyle{plain}
  \newtheorem{lemma}[thm]{Lemma}

  \theoremstyle{plain}
  \newtheorem*{lemma*}{Lemma}

  \theoremstyle{definition}
  \newtheorem{defn}[thm]{Definition}

  \theoremstyle{plain}
  \newtheorem{cor}[thm]{Corollary}

  \theoremstyle{plain}
  \newtheorem*{cor*}{Corollary}

  \theoremstyle{remark}
  \newtheorem*{rem*}{Remark}

  \theoremstyle{plain}
  
\date{\today}
\eqnobysec   
\newcommand{\Cov}{\mathop{\mathrm{Cov}}\nolimits}
\newcommand{\Var}{\mathop{\mathrm{Var}}\nolimits}

\renewcommand{\Im}{\mathop{\mathrm{Im}}\nolimits}

\newcommand{\GS}{\mathcal{GS}_d(\lambda)}
\newcommand{\be}{\begin{eqnarray}}
\newcommand{\ee}{\end{eqnarray}}

\begin{document}

\title{Gaussian Waves on the Regular Tree}
\author{Yehonatan Elon}
\address {Department of Physics of Complex Systems,\\ The Weizmann
Institute of Science, 76100 Rehovot, Israel}
\begin{abstract}
We consider the family of real (generalized) eigenfunctions of the
adjacency operator on $T_d$ - the $d$-regular tree. We show the
existence of a unique invariant Gaussian process on the ensemble and
derive explicitly its
covariance operator.\\
We investigate the typical structure of level sets of the process.
In particular we show that the entropic repulsion of the level sets
is uniformly bounded and prove the existence of a critical
threshold, above which the level sets are all of finite cardinality
and below it an infinite component appears almost surely.
\end{abstract}

\section{Introduction and main results}\label{sec:Introduction}
The regular tree $T_d$ (also known as the Bethe-lattice), is a
connected, cycle-free, infinite graph where each vertex is connected
to $d$ neighbors. In an abuse of notation, we shall use the notation
$T_d$ for both the graph and the set of its vertices.\\
For a function $f:T_d\rightarrow\mathbb{R}$ the adjacency operator
$A$ acts on the components of $f$ by
 \be\nonumber
  (A f)(v)=\sum_{|v-v'|=1}f(v')
 \ee
where $v,v'\in T_d$ and $|v-v'|$ is the distance in $T_d$ from $v$
to $v'$.\\
The spectrum \cite{Cartier} of $A$ is absolutely continuous,
supported on the interval
 \be\nonumber
  \sigma(T_d)=[-2\sqrt{d-1},2\sqrt{d-1}],
 \ee
with a spectral density, given by
 \be\nonumber
  \rho(\lambda)=\frac{d}{2\pi}\frac{\sqrt{4(d-1)-\lambda^2}}{d^2-\lambda^2}
  \mathcal{I}_{\lambda\in\sigma(T_d)}
 \ee
where $\mathcal{I}$ is the indicator function.  For a given
$\lambda\in\sigma(T_d)$, a function $\psi:T_d\rightarrow\mathbb{C}$
will be referred as a \textit{wave} (or a generalized eigenfunction)
if $(A-\lambda I)\psi=0$, where $I$ is the identity
operator on $T_d$.\\
In this paper, we investigate the existence and properties of the
following Gaussian process on $T_d$:
\begin{thm}\label{thm:Ensemble}
 For every $d\ge3$ and $\lambda\in\sigma(T_d)$, there exists a unique
 random process $\GS=\{\Omega,\mu\}$, associating $\forall\omega\in\Omega$
 a function $\psi_\omega:T_d\rightarrow\mathbb{R}$ with the
 following properties:
 \begin{enumerate}
  \item for almost every $\omega\in\Omega$,
   $(A-\lambda I)\psi_\omega=0$ .
  \item $\mu$ is a Gaussian measure, where $\forall v\in T_d$, the marginal
   variance $\Var(\psi_\omega(v))=1$ .
  \item for every automorphism $\Phi:T_d\rightarrow T_d$ and vertices $v,v'\in T_d$,
    \be\nonumber
     \mathbb{E}(\psi_\omega(v)\psi_\omega(v'))=\mathbb{E}(\psi_\omega(\Phi(v))\psi_\omega(\Phi(v')))
    \ee
 \end{enumerate}
\end{thm}
Gaussian processes are frequently used in various branches of
physics, such as semi-classical analysis \cite{Berry77}, optics
\cite{Goodman} or cosmology \cite{Liddle} (to list only a few) -
usually as a heuristic model to study systems with random
perturbations. In particular, $\GS$ was conjectured in
\cite{Elon08}, as a limiting process for the
distribution of eigenvectors of random regular graphs.\\
Theorem \ref{thm:Ensemble} will be established in section
\ref{sec:construction} where we investigate the properties of $\GS$,
by calculating the covariance operator explicitly. In addition, we
prove that the process has a Markov property, in a sense that will
be defined in theorem \ref{thm:Markov}.\\
We will also characterize the structure of a typical realization
$\psi_\omega$, by considering its level sets:\\
For a function $f:T_d\rightarrow\mathbb{R}$ and
$\alpha\in\mathbb{R}$, we define the induced subgraph
$T_\alpha(f)\subset T_d$, by keeping only vertices above the
threshold $\alpha$:
   \be\nonumber
    T_\alpha(f)=\{v\in T_d, f(v)>\alpha\}
   \ee
and define the \textit{$\alpha$-level sets} of $f$, to be the
connected components of $T_\alpha(f)$.\\
One aspect that will be investigated is the \textit{entropic
repulsion} induced by the process, namely the distribution of
$\psi_\omega(v)$, conditioned on the diameter of the $\alpha$-level
set, containing $v$. Let
 \be\nonumber
  V=\{v_j\}_{j=1}^n\subset T_d
 \ee
be a simple path of length $n-1$, so that $\forall 1<j<n$,
$|v_{j+1}-v_j|=|v_j-v_{j-1}|=1$ but $v_{j-1}\ne v_{j+1}$.\\
For a given $\alpha\in\mathbb{R}$ and $n\in\mathbb{N}$, we define
the conditional probability
 \be\nonumber
  \mathbb{P}^{+,\alpha}_n(\cdot)=\mathbb{P}(\cdot|\forall v\in V,\psi_\omega(v)>\alpha)
 \ee
and argue the following:
 \begin{thm}\label{thm:EntropicRepulsion}
  $\forall\lambda\in\sigma(T_d)$ and $\alpha\in\mathbb{R}$, there
  exist $\psi_0(\lambda,\alpha)<\infty$ and $0<c_1(\lambda),c_2(\lambda)<\infty$,
  so that $\forall n\in\mathbb{N}, 1\le k\le n$ and $x>\psi_0$
  \be\nonumber
   \mathbb{P}^{+,\alpha}_n\left(\psi_\omega(v_k)\ge x\right)<c_1\e^{-c_2 x^2}
  \ee
 \end{thm}
A natural question which arise, when considering the structure of
the $\alpha$-level sets of $\GS$, is related to the existence, or
the absence, of an infinite level set and the transition between the
two regimes. This question is answered by the following theorem:
 \begin{thm}\label{Thm:PT}
  $\forall\lambda\in\sigma(T_d)$, there exists an $\alpha_c\in\mathbb{R}$ so
  that for almost every realization $\psi_\omega\in\GS$, $T_\alpha(\psi_\omega)$
  has an infinite component for $\alpha<\alpha_c$, but only finite components for
  $\alpha>\alpha_c$.
 \end{thm}
The theorem is proved in section \ref{sec:PT}, following a
'quasi-Bernoulli' criterion, introduced by Lyons in \cite{Lyons},
for random percolation processes on tree graphs.
\subsection{Relations with previous results}\label{subsec:relations}
Gaussian waves on $\mathbb{R}^n$ were first suggested in
\cite{Berry77} as a model for the limiting behavior of
eigenfunctions of chaotic systems. While the model is not supported
by any rigorous derivation, it was found consistent with some
numerical observations, such as \cite{Blum,Elon,Bogomolny02}.\\
In \cite{Elon08}, a modified 'random waves' model was introduced, in
order to describe the statistics of adjacency eigenvectors of the
ensemble $G(n,d)$, consisted of all $d-$regular graphs on $n$
vertices and equipped with the uniform measure.\\
The ensemble $G(n,d)$ serves frequently as a convenient model for
random expander graphs (for a review, consider
\cite{LinialExpanders}). Eigenvectors of such graphs are used in
various algorithms (e.g. \cite{Shi,Pothen,Coifman}), however not
much is known about their characteristics.\\
$G(n,d)$ graphs have drawn recently a considerable attention in the
physical community as a plausible 'toy-model' for generic chaotic
systems. In \cite{Rudnick} it was claimed, based on numerical
simulations, that in the limit $n\rightarrow\infty$, the local level
distribution of such graphs follows the predictions of the GOE
random matrices ensemble, as believed to hold for chaotic billiards
as well \cite{Bohigas}. This result was recently strengthened by an
analytic derivation \cite{OrenI} of the $2$-levels form-factor
asymptotics. The spectral properties of $G(n,d)$ were also suggested
in \cite{Aizenman06} as a natural finite dimensional model for the
regular tree, in the context of Anderson (de-) localization.\\
The model conjectured in \cite{Elon08}, has to do with these three
aspects: it relates $G(n,d)$ graphs to an additional universality
class associated with chaotic behavior, it predicts that the
eigenvectors of such graphs are extended (corresponding to the
appearance of an ac spectrum in the corresponding lattice). Lastly,
it predicts with a high accuracy variate statistical properties of
the eigenvectors, which are of interest - for example, the nodal
domains statistics of such graphs, which was measured in
\cite{Dekel}, but have not found any explanation.\\
In this paper we provide a rigorous construction of the Gaussian
waves model on $T_d$, which is a first step towards the analysis of
the relations between the process $\GS$ and the eigenvectors of a
random $G(n,d)$ graph.\\
The statistics of level sets of Gaussian random waves in
$\mathbb{R}^2$ and specifically their nodal sets was measured and
characterized in \cite{Blum,Bogomolny07}. The observed statistics
found an intriguing explanation in \cite{Bogomolny02}, where it was
conjectured that the nodal statistics of $2$ dimensional Gaussian
waves can be approximated by a critical (non correlated) percolation
model.\\
The last section of the current work provides a rigorous proof for
the critical behavior of the Gaussian waves model in $T_d$. However,
the characters of the transition are different then the ones of
uncorrelated percolation.

\section{Properties of the process $\GS$}\label{sec:construction}
As a Gaussian process is characterized by its covariance operator,
theorem \ref{thm:Ensemble} will follow from lemmas
\ref{lem:existence} and \ref{lem:uniqueness}, where we prove the
existence of the described process and the uniqueness of its
covariance.\\
We denote by $\delta_v\in l^2(T_d)$ the indicator function supported
at $v\in T_d$. We define for $z\in\mathbb{C}\setminus\sigma(T_d)$
the resolvent operator $R(A,z):T_d\rightarrow T_d$ by
 \be\nonumber
  R(A,z)=(A-zI)^{-1}
 \ee
We will also make use of the Chebyshev Polynomials of the second
kind, defined as
 \be\label{eq:chebyshev}
  U_n(x)=\frac{\sin\left((n+1)\cos^{-1}(x)\right)}{\sin\left(\cos^{-1}(x)\right)}
 \ee
and follow the convention $U_{-1}(x)=U_1(x)$.
 \begin{lemma}{\label{lem:existence}}
  For every $\lambda\in\sigma(T_d)$, the Gaussian process determined by the covariance
  operator
  $\Cov(\psi_\omega(v),\psi_\omega(v'))=\langle\delta_v,C_\lambda\delta_{v'}\rangle$,
  where
  \be\label{eq:CovDef}
   \langle\delta_v,C_\lambda\delta_{v'}\rangle=\lim_{\epsilon\rightarrow0^+}\frac{\Im\left\langle\delta_v,R(A,\lambda+i\epsilon)\delta_{v'}\right\rangle}
   {\Im\left\langle\delta_v,R(A,\lambda+i\epsilon)\delta_{v}\right\rangle}
  \ee
 Is consistent with the requirements of theorem \ref{thm:Ensemble}
 \end{lemma}
 \begin{proof}
The existence of the limit appearing in equation \ref{eq:CovDef} can
be verified for $\lambda\in\sigma(T_d)$, by considering the spectral
representation of the resolvent and recalling the smoothness of the
spectral density (see, for example sections 1.3, 1.4 of
\cite{Carmona}). As
  \be\label{eq:CovDecomposition}
   \Im(A-(\lambda+i\epsilon)I)^{-1}=\epsilon((A-\lambda I)^2+\epsilon^2I)^{-1}
  \ee
The resolvent is positive definite $\forall\epsilon>0$, therefore
equation \ref{eq:CovDef} defines an appropriate covariance operator.\\
Following equation \ref{eq:CovDef}, $\Var(\psi_\omega(v))=1$.
Moreover, as the resolvent is invariant under any automorphism
$\Phi:T_d\rightarrow T_d$ requirements $(ii)$ and $(iii)$ of
theorem \ref{thm:Ensemble} are satisfied.\\
Lastly, let $\psi_\omega$ be a random realization of the Gaussian
 process generated by $C_\lambda$. As the law of
 $\psi_\omega$ is invariant under reflections,
  \be\nonumber
   \mathbb{E}(\langle\delta_v,(A-\lambda I)\psi_\omega\rangle)=0\ .
  \ee
In addition, by the invariance of the measure,
 \be\nonumber
  \mathbb{E}(\langle\delta_v,(A-\lambda I)\psi_\omega\rangle^2)=\\\nonumber
   (d+\lambda^2)\mathbb{E}(\psi_\omega(v)^2)-2d\lambda\mathbb{E}(\psi_\omega(v)\psi_\omega(v'))_{|v-v'|=1}
   +d(d-1)\mathbb{E}(\psi_\omega(v)\psi_\omega(v''))_{|v-v''|=2}=\\\nonumber
  \mathbb{E}(\langle\delta_v,(A-\lambda I)^2\psi_\omega\rangle
  \langle\psi_\omega,\delta_v\rangle)=\\\nonumber
   \langle\delta_v (A-\lambda I)^2, C_\lambda\delta_{v}\rangle=\\\nonumber
  \lim_{\epsilon\rightarrow0}\frac{\langle\delta_v,(A-\lambda I)^2((A-\lambda I)^2+\epsilon^2 I)^{-1}\delta_v\rangle}
  {\langle\delta_v,((A-\lambda I)^2+\epsilon^2I)^{-1}\delta_v\rangle}=0
 \ee
where in the first and second step we have expanded the bilinear
form into elements (bearing in mind that
$\mathbb{E}(\psi_\omega(v)\psi_\omega(v'))$ depends only on
$|v-v'|$) and recollected it; In the third step we have followed the
definition of the covariance operatot:
$C_\lambda=\mathbb{E}(\psi_\omega\psi_\omega^T)$ and in the fourth,
we have followed equations \ref{eq:CovDef}, \ref{eq:CovDecomposition}.\\
 As $\langle\delta_v,(A-\lambda I)\psi_\omega\rangle$ is a Gaussian random variable, with zero mean
 and variance, it equals zero almost surely, establishing by that requirement $(i)$ of theorem
 \ref{thm:Ensemble}
 \end{proof}
Note that the last proof relies only on the smoothness of the
spectral density of $A$, and the invariance of the process.
Therefore Gaussian wave models can be generated by the covariance
operator \ref{eq:CovDef} for broader classes of conducting graphs.

In order to find an explicit expression for the covariance of the
process $\GS$ for a given $\lambda\in\sigma(T_d)$, we would like to
introduce the function \cite{Cartier,Brooks}
$\phi^{(\lambda)}:\mathbb{N}\rightarrow\mathbb{R}$ defined as
 \be\label{eq:CofE}
  \phi^{(\lambda)}(n)=(d-1)^{-n/2}\left(\frac{d-1}{d}U_n\left(
  \textstyle{\frac{\lambda}{2\sqrt{d-1}}}\right)-\frac1dU_{n-2}\left(\textstyle{\frac{\lambda}{2\sqrt{d-1}}}\right)\right)
 \ee
and state the following:
 \begin{lemma}\label{lem:uniqueness}
Let $\GS=\{\Omega,\mu\}$ be a random Gaussian process, consistent
with the requirements of theorem \ref{thm:Ensemble}. Then, the
covariance of $\GS$ is given by
  \be\nonumber
   \Cov(\psi_\omega(v),\psi_\omega(v'))=\phi^{(\lambda)}(|v-v'|)
  \ee
 \end{lemma}
\begin{proof}
We begin by considering a general property of waves on $T_d$.\\
Let $f:T_d\rightarrow\mathbb{R}$, so that $(A-\lambda I)f=0$. For a
given $v\in T_d$ and $k\in\mathbb{N}$, denote the sphere of radius
$k$ around $v$ by
 \be\nonumber
  \Lambda_k(v)=\{v'\in T_d,|v'-v|=k\}
 \ee
and define
  \be\nonumber
   S_k(f,v)=\sum_{v'\in\Lambda_k(v)}f(v')
  \ee
to be the sum of $f$ over the $k-$sphere. As $(A-\lambda I)f=0$ and
since for $k\ge2$, every vertex in the $(k-1)^{th}$ sphere has $d-1$
neighbors in the $k^{th}$ sphere , we get that:
 \be\label{eq:recursionS}
  S_0=f(v)\quad,\quad S_1=\lambda f(v)\\\nonumber
  \lambda S_k=(d-1)S_{k-1}+S_{k+1}\quad\textrm{for }k\ge2
 \ee
Recalling that Chebyshev polynomials are related by the recursion
relation
 \be\nonumber
  2xU_k(x)=U_{k-1}(x)+U_{k+1}(x)
 \ee
one can verify that
 \be\label{eq:Ak}
  S_k(f,v)=|\Lambda_k|\phi^{(\lambda)}(k)\cdot f(v)
 \ee
is the (unique) solution to \ref{eq:recursionS}.\\
Now, assume that a process $\GS=\{\Omega,\mu\}$ follows the
requirements made in theorem \ref{thm:Ensemble}. Then,
 \be\label{eq:CovOfPsi}
   \mathbb{E}(\psi_\omega(v)\cdot\psi_\omega(v'))|_{|v-v'|=k}&=&
   \mathbb{E}(\psi_\omega(v)\cdot\frac1{|\Lambda_k|}S_k(\psi_\omega,v))\\\nonumber
   &=&\phi^{(\lambda)}(k)\mathbb{E}(\psi_\omega^2(v))\\\nonumber
   &=&\phi^{(\lambda)}(k)
  \ee
Where we have followed properties $(iii)$, $(i)$ and $(ii)$ of the
process $\GS$ respectively.
 \end{proof}

In the rest of this paper, we will be often interested in the
restriction of the process $\GS$ to finite subsets of $T_d$. For
this reason we would like to introduce the following notation:\\
For a set $V=\{v_i\}_{i=1}^n\subset T_d$ we denote
 \be\nonumber
  \psi_\omega(V)=(\psi_\omega(v_1),...,\psi_\omega(v_n))\quad,\quad
  d\psi_\omega(V)=\prod_{i=1}^n d\psi_\omega(v_i)\ .
 \ee
The density of the measure on $\GS$ will be denoted by
 \be\nonumber
  d\mu(\psi_\omega(V))=p(\psi_\omega(V))d\psi_\omega(V)
 \ee

The adjacency operator is local, i.e. it contains only nearest
neighbors interactions. For tree graphs, such as $T_d$, this
property has the following consequence:\\
Consider two adjacent vertices $v_1,v_2\in T_d$, define the
partition of $T_d$ into
 \be\nonumber
  T^{(1)}(v_1,v_2)=\left\{v\in T_d,|v-v_1|<|v-v_2|\right\}\\\nonumber
  T^{(2)}(v_1,v_2)=\left\{v\in T_d,|v-v_1|>|v-v_2|\right\}
 \ee
(see figure \ref{fig:T1T2}) and Let $V_1\subset T^{(1)}\setminus
v_1$ and $V_2\subset T^{(2)}\setminus v_2$ be finite subsets of the
two subgraphs.\\
\begin{figure}[h]
\centering
 \scalebox{0.6}{\includegraphics{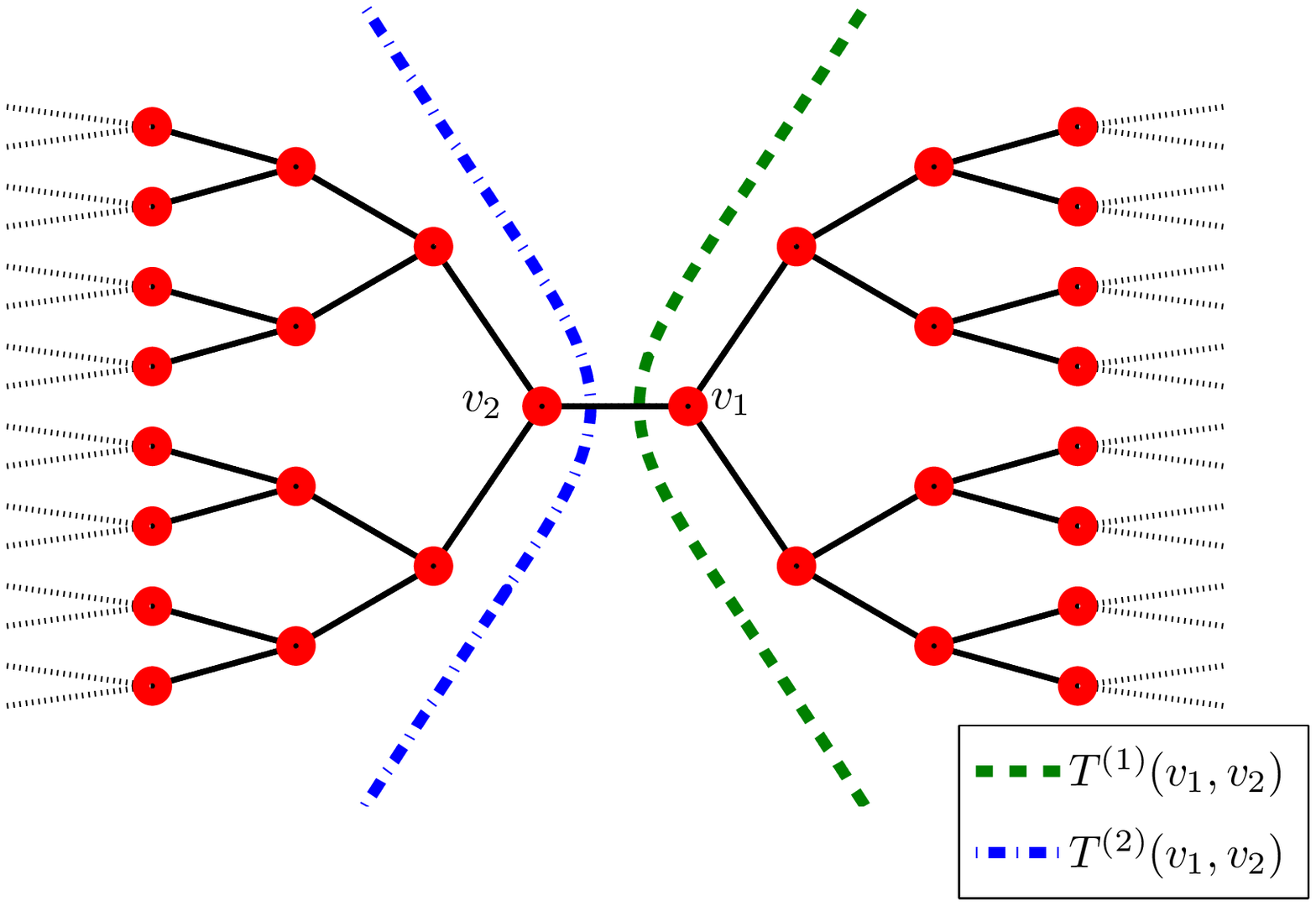}}
     \caption{}
    \label{fig:T1T2}
\end{figure}
For a given $\lambda\in\sigma(T_d)$ and $X\in\mathbb{R}^{|V_1|+2}$,
consider the family of waves on $T_d$, where we fix the value of the
function on $V_0$ and $V_1$ to $X$:
 \be\nonumber
  F=\left\{f:T_d\rightarrow\mathbb{R},(A-\lambda I)f=0,f(V_0\cup V_1)=X\right\}\ .
 \ee
Due to the constraints which are imposed by the adjacency operator,
fixing $f(V_0)$ might impose constraints on $f(V_2)$
which must be satisfied $\forall f\in F$. However, note that by
fixing $f(V_0)$, the adjacency operator does not mix vertices from
$V_1$ and $V_2$. As a result, the constraints on $f(V_2)$, imposed
by fixing $f(V_0\cup V_1)$ are identical to the one imposed by
fixing $f(V_0)$ alone.\\
This property is inherited by the process $\GS$ in the following
sense:
 \begin{thm}\label{thm:Markov}
Let $V_0=\{v_1,v_2\}\subset T_d$ , so that $|v_1-v_2|=1$,
$V_1\subset T^{(1)}(v_1,v_2)$ and $V_2\subset T^{(2)}(v_1,v_2)$.
Then the distribution of $\psi_\omega(V_2)$ conditioned on
$\psi_\omega(V_0)$ is independent of $\psi_\omega(V_1)$:
   \be\nonumber
    p\left(\psi_\omega(V_2)|\psi_\omega(V_0\cup V_1)\right)=
    p\left(\psi_\omega(V_2)|\psi_\omega(V_0)\right)
   \ee
 \end{thm}
 \begin{proof}
As $\psi_\omega(V_2)$ is a Gaussian random vector, it is enough to
show that
 \be\nonumber
  \mathbb{E}(\psi_\omega(V_2)|\psi_\omega(V_0\cup V_1))=
  \mathbb{E}(\psi_\omega(V_2)|\psi_\omega(V_0))
 \ee
and that
 \be\nonumber
  C_\lambda(\psi_\omega(V_2)|\psi_\omega(V_0\cup V_1))=
  C_\lambda(\psi_\omega(V_2)|\psi_\omega(V_0))
 \ee
where $C_\lambda$ is the conditional covariance operator.\\
To do so, we set $n_1=|V_1|$, $n_2=|V_2|$ and denote by
$V=\{v_k\}_{k=1}^{2+n_1+n_2}$ the union of $V_0,V_1,V_2$. We
consider first the case where the adjacency operator do not impose
constraints on $\psi_\omega(V)$, so that the covariance matrix
$(C_\lambda(V))_{ml}=\phi^{(\lambda)}(|v_m-v_l|)$ is strictly
positive. $C_\lambda(V)$ can be written in the next blocks form:
 \be\nonumber
  C_\lambda(V)=\left(\begin{array}{cc}C_{(11)}&C_{(12)}\\C_{(21)}&C_{(22)}\end{array}\right)
 \ee
Where $C_{(11)}$ is the $(n_1+2)\times(n_1+2)$ covariance matrix of
$\psi_\omega(V_0\cup V_1)$ and $C_{(22)}$ is the $n_2\times
n_2$ covariance matrix of $\psi_\omega(V_2)$.\\
We set, in a similar fashion to the proof of lemma
\ref{lem:uniqueness},
 \be\nonumber
  \tilde\Lambda_k(v_2)=\{v\in T^{(2)},|v-v_2|=k\}\\\nonumber
  \tilde S_k(\psi_\omega,v_2)=\sum_{v\in\tilde\Lambda_k(v_2)}\psi_\omega(v)
 \ee
As $(A-\lambda I)\psi_\omega=0$, the $\tilde S_k$'s are determined
by the recursion relation
 \be\nonumber
  \tilde S_0=\psi_\omega(v_2),\quad
  \tilde S_1=\lambda\psi_\omega(v_2)-\psi_\omega(v_1)\\\nonumber
  \lambda \tilde S_k=(d-1)\tilde S_{k-1}+\tilde S_{k+1}
 \ee
Therefore $\forall k\in\mathbb{N},\ \tilde S_k$ is determined by
$\psi_\omega(V_0)$. By the invariance of the process $\GS$, we
obtain that $\forall k\ge0$ and $v\in\tilde\Lambda_k$,
 \be\nonumber
  \mathbb{E}(\psi_\omega(v)|\psi_\omega(V_0),\psi_\omega(V_1))=\tilde S_k/|\tilde\Lambda_k|
 \ee
As a result, since $V_2\subset\bigcup_k S_k$, we find that
 \be\nonumber
  \mathbb{E}(\psi_\omega(V_2)|\psi_\omega(V_0),\psi_\omega(V_1))=\mathbb{E}(\psi_\omega(V_2)|\psi_\omega(V_0))\ .
 \ee
The conditional expectation and covariance operator of
$\psi_\omega(V_2)$ are given by the formulae
 \be\label{eq:condEC}
  \mathbb{E}(\psi_\omega(V_2)|\psi_\omega(V_0),\psi_\omega(V_1))=
     C_{(21)}C_{(11)}^{-1}\psi_\omega(V_0\cup V_1)\\\nonumber
  \left.C_\lambda(V_2)\right|_{\psi_\omega(V_0),\psi_\omega(V_1)}=
     C_{(22)}-C_{(21)}C_{(11)}^{-1}C_{(12)}
 \ee
Therefore, since
$\mathbb{E}(\psi_\omega(V_2)|\psi_\omega(V_0),\psi_\omega(V_1))$ is
independent of $\psi_\omega(V_1)$, $(C_{(21)}C_{(11)}^{-1})_{ij}$
must vanish $\forall j>2$ and $(C_{(21)}C_{(11)}^{-1})_{ij}$ is
independent of the set $V_1$ for $i=1,2$. As a result,
$C_{(21)}C_{(11)}^{-1}C_{(12)}$ is independent of $V_1$ as well,
therefore
 \be\nonumber
  \left.C_\lambda(V_2)\right|_{\psi_\omega(V_0),\psi_\omega(V_1)}=
  \left.C_\lambda(V_2)\right|_{\psi_\omega(V_0)}
 \ee
Establishing by that the suggested independence.\\
As was noted above, if the adjacency operator does impose constrains
on the distribution of $\psi_\omega(V)$, these constraints can be
decoupled into separate constraints on $\psi_\omega(V_1)$ and
$\psi_\omega(V_2)$. Therefore, there exists a partition of $V_1,\
V_2$ into a free and constrained subsets: $V_1=V_1^{F}\cup V_1^{C}$
and $V_2=V_2^{F}\cup V_2^{C}$, so that $C_\lambda(V_0\cup V_1^F\cup
V_2^F)$ is strictly positive, while $\psi_\omega(V_1^C)$ and
$\psi_\omega(V_2^C)$ are determined uniquely by $\psi_\omega(V_0\cup
V_1^F)$ and $\psi_\omega(V_0\cup V_2^F)$ correspondingly.\\
Therefore, from the proof to the unconstrained case we obtain that
 \be\nonumber
  p\left(\psi_\omega(V_2^F)|\psi_\omega(V_0\cup V_1)\right)=
  p\left(\psi_\omega(V_2^F)|\psi_\omega(V_0\cup V_1^F)\right)=
  p\left(\psi_\omega(V_2^F)|\psi_\omega(V_0)\right)
 \ee
As $\psi_\omega(V_2^C)$ is uniquely determined by
$\psi_\omega(V_2^F)$, the theorem follows.
 \end{proof}

\section{Distribution of level sets}\label{sec:entropic_repulsion}
As was suggested in section \ref{sec:Introduction}, the
$\alpha$-level sets of $\psi_\omega\in\GS$, can be naturally related
to a random process $\{\Omega,\mathbb{P}_\alpha\}$ on the Bethe
lattice, associating $\forall\omega\in\Omega$, an induced subgraph
$T_\alpha(\psi_\omega)\subset T_d$, according to the rule:
 \be\nonumber
  T_\alpha(\psi_\omega)=\{v\in T_d,\psi_\omega(v)>\alpha\}\ .
 \ee
For $\alpha\in\mathbb{R}$, $\psi_\omega\in\GS$ and $v\in T_d$, we
set $C_\omega^{\alpha}(v)$ to denote the connected component of $v$
in
$T_\alpha(\psi_\omega)$.\\
In the following, we will consider the conditional distribution of
$\psi_\omega(v)$, where we condition on the diameter of the
$\alpha$-level set which contains $v$.\\
Let $V=\{v_j\}_{j=1}^n$ be a simple path in $T_d$. For a given
$\alpha\in\mathbb{R}$ we set
 \be\nonumber
  \Omega^{+,\alpha}_n=\{\omega\in\Omega,v_1\in C_\omega^{\alpha}(v_n)\}
 \ee
as the restriction of the sample space $\Omega$ to events in which
$V$ is contained in an $\alpha$-level set. Similarly, we use the
symbols
 \be\nonumber\fl
  \mathbb{P}^{+,\alpha}_n(\cdot)=\mathbb{P}(\cdot|v_1\in C_\omega^{\alpha}(v_n))\quad,\quad
  \mathbb{E}^{+,\alpha}_n(\cdot)=\mathbb{E}(\cdot|v_1\in C_\omega^{\alpha}(v_n))\quad,\quad
  p^{+,\alpha}_n(\cdot)=p(\cdot|v_1\in C_\omega^{\alpha}(v_n))
 \ee
to denote probabilities, expectations and densities, conditioned on
the event $v_1\in C_\omega^{\alpha}(v_n)$.\\
The main result of this section is theorem
\ref{thm:EntropicRepulsion}. The proof of the theorem will follow
the next lines:\\
First, we calculate in lemma \ref{lem:CondDist} the probability
density $p_n^+(\psi_\omega(v)|\psi_\omega(V\setminus v))$ and find
that it is concentrated around a linear combination of
$\{\psi_\omega(v'\in V),|v-v'|\le2\}$, with a
bounded variance and Gaussian tails.\\
Next, in lemma \ref{lem:bound_max}, we observe that above some
finite threshold $\psi_0(\lambda,\alpha)<\infty$, the suggested
linear combination becomes convex. Therefore, the probability to
find that $\psi_\omega(v)>x$ decays rapidly for $x>\psi_0$, unless
$\psi_\omega(v)$ is significantly smaller then the average of its
neighbors. Finally, we show that the convexity of the distribution,
results in the concentration of $\psi_\omega(v)$, establishing by
that theorem \ref{thm:EntropicRepulsion}.
 \begin{lemma}\label{lem:CondDist}
Let $V=\{v_j\}_{j=1}^n\subset T_d$ be a simple path. Then
$\forall\lambda\in\sigma(T_d), \alpha\in\mathbb{R}$ and $v_k\in V$,
there exist constants
$a_{k\pm1}(\lambda,\alpha),a_{k\pm2}(\lambda,\alpha)$ and
$\sigma_k^2<1$, so that
 \be\label{eq:condDist}
  p^{+,\alpha}_n(\psi_\omega(v_k)=x|\psi_\omega(V\setminus v_k))=\frac1{\mathcal{Z}}
  \exp\left(-\frac{(x-E_k)^2}{2\sigma_k^2}\right) \mathcal{I}_{x>\alpha}
 \ee
where $\mathcal{I}$ is the indicator function,
$\mathcal{Z}=\int_{\alpha}^{\infty}dy\exp\left(-\frac
{(y-E_k)^2}{2\sigma_k^2}\right)$ and
 \be\nonumber
  E_k(\omega)=\sum_{|k-k'|\le2}a_{k'}(\lambda,\alpha)\psi_\omega(v_{k'})
 \ee
 \end{lemma}
 \begin{proof}
First, note that
 \be\nonumber
  p^{+,\alpha}_n(\psi_\omega(v_k)=x|\psi_\omega(V\setminus v_k))=
  \frac{\mathcal{I}_{x>\alpha}}{\mathcal{Z}}\cdot p(\psi_\omega(v_k)=x|\psi_\omega(V\setminus v_k))
 \ee
Therefore, as $p(\psi_\omega(v_k)|\psi_\omega(V\setminus v_k))$ is
Gaussian, equation \ref{eq:condDist} follows with
 \be\nonumber
  E_k(\omega)=\mathbb{E}(\psi_\omega(v_k)|\psi_\omega(V\setminus v_k)),\quad
  \sigma_k^2=\Var(\psi_\omega(v_k)|\psi_\omega(V\setminus v_k))
 \ee
Next, according to theorem \ref{thm:Markov}, we get that
 \be\nonumber
  p(\psi_\omega(v_k)=x|\psi_\omega(V\setminus v_k))=p(\psi_\omega(v_k)=x|\psi_\omega(\tilde{V}))
 \ee
where $\tilde{V}=(v_{k-2},v_{k-1},v_{k+1},v_{k+2})$. Therefore,
following formula \ref{eq:condEC}, we find that
$\Var(\psi_\omega(v_k)|\psi_\omega(V\setminus
v_k))<\Var(\psi_\omega(v_k))=1$ and that
$\mathbb{E}(\psi_\omega(v_k)|\psi_\omega(V\setminus v_k))$ is a
linear combination of $\psi_\omega(\tilde V)$, establishing by that
the lemma.\\
By a straight forward calculation (which involves the inversion of a
$4\times4$ matrix), we obtain that the conditional expectation of
$\psi_\omega(v_k)$ is given by
 \be\label{eq:ExpK}
  E_1(\omega)=
   \frac{\lambda\psi_\omega(v_2)-\psi_\omega(v_3)}{d-1}\\\nonumber
  E_2(\omega)=
   \frac{(d-1)\lambda\psi_\omega(v_1)+d\lambda\psi_\omega(v_3)-(d-1)\psi_\omega(v_4)}{\lambda^2+(d-1)^2}\\\nonumber
  E_k(\omega)=
   a_1(\lambda)\frac{\psi_\omega(v_{k-1})+\psi_\omega(v_{k+1})}2-a_2(\lambda)\frac{\psi_\omega(v_{k-2})+\psi_\omega(v_{k+2})}2
 \ee
for $2<k<n-1$, where
 \be\nonumber
  a_1(\lambda)=\frac{2d\lambda}{\lambda^2+(d-1)^2+1}\quad,\quad
  a_2(\lambda)=\frac{2(d-1)}{\lambda^2+(d-1)^2+1}\ .
 \ee
$E_{n-1},E_n$ are obtained from $E_1,E_0$ by a reindexation of $V$.
 \end{proof}
\noindent According to the last lemma,
$p^{+,\alpha}_n(\psi_\omega(v_k)|\psi_\omega(V\setminus v_k))$ is
concentrated with Gaussian tails near its expectation value, which
is bounded from above by
 \be\label{eq:CondExp}
  \mathbb{E}^{+,\alpha}_n(\psi_\omega(v_k)|\psi_\omega(V\setminus v_k))<\max(E_k(\omega),\alpha)+1
 \ee
An important observation, which will have a significant role in the
proof of theorem \ref{thm:EntropicRepulsion}, is the convexity of
$E_k(\omega)$. Note that $\forall\lambda\in\sigma(T_d)$ and $1\le
k\le n$, the sum of coefficients appearing in equation \ref{eq:ExpK}
is smaller than one, implying that
$\mathbb{E}^{+,\alpha}_n(\psi_\omega(v_k))$ cannot exceed
significantly the average of its neighbors. Introducing the notation
$\psi_\omega(v_{-1})\equiv\psi_\omega(v_1),\psi_\omega(v_{n+1})\equiv\psi_\omega(v_{n-1})$,
the next lemma follows:
 \begin{lemma}\label{lem:bound_max}
$\forall\lambda\in\sigma(T_d)$ and $\alpha\in\mathbb{R}$,
$\exists\psi_0(\lambda,\alpha)<\infty$ and
$c_2(\lambda),c_3(\lambda)>0$, so that $\forall n\in\mathbb{N}, 1\le
k\le n$ and $x>\psi_0$
 \be\nonumber
  \mathbb{P}^{+,\alpha}_n\left(\psi_\omega(v_k)>x\wedge\psi_\omega(v_k)>\frac{1-c_3}{2}(\psi_\omega(v_{k-1})+\psi_\omega(v_{k+1}))\right)
  <\frac12\e^{-c_2 x^2}
 \ee
 \end{lemma}
For the sake of brevity and in order to avoid messy calculations, we
consider here only the case where $\lambda<d-\sqrt{2(d-1)}$ and
$2<k<n-1$, where the completion of the proof is postponed to
\ref{app:bound_max}. Note that as
$\sigma(T_d)=[-2\sqrt{d-1},2\sqrt{d-1}]$, the following proof is
incomplete only for $d\le10$, where $d-\sqrt{2(d-1)}<2\sqrt{d-1}$.
 \begin{proof}
\textit{(partial)} According to equation \ref{eq:ExpK} and as
$\forall\omega\in\Omega^+,\psi_\omega(v_{k\pm2})>\alpha$, we observe
that
 \be\nonumber
  E_k(\omega)<\frac{a_1(\lambda)}2(\psi_\omega(v_{k-1})+\psi_\omega(v_{k+1}))+a_2(\lambda)|\alpha|
 \ee
Note that for $\lambda<d-\sqrt{2(d-1)}$, $a_1(\lambda)<1$.
Therefore, either $E_k(\omega)$ is bounded from above, or
$E_k(\omega)<(\psi_\omega(v_{k-1})+\psi_\omega(v_{k+1}))/2$. Setting
 \be\nonumber
  c_3(\lambda)=\frac{1-a_1(\lambda)\mathcal{I}_{\lambda>0}}3\quad,\quad
  \psi_0(\lambda,\alpha)=\max\left\{\frac{a_2(\lambda)|\alpha|}{c_3(\lambda)}(1-c_3),2(|\alpha|+1)\right\}
 \ee
we obtain that $c_3>0$ and
 \be\nonumber
  E_k(\omega)<\frac{1-3c_3}2(\psi_\omega(v_{k-1})+\psi_\omega(v_{k+1}))+\frac{c_3}{1-c_3}\psi_0,
 \ee
therefore, according to equation \ref{eq:CondExp}
 \be\nonumber
  \phantom{<}\mathbb{P}^{+,\alpha}_n\left(\psi_\omega(v_k)>x\wedge\psi_\omega(v_k)>\frac{1-c_3}{2}(\psi_\omega(v_{k-1})+\psi_\omega(v_{k+1}))\right)\\\nonumber
  <\mathbb{P}^{+,\alpha}_n\left(\psi_\omega(v_k)-\max(\alpha+1,E_k(\omega)+1)>\frac{c_3}{1-c_3}x\right)\\\nonumber
  <\frac12e^{-c_2x^2}
 \ee
where $c_2={c_3^2}/{2\sigma_k^2(1-c_3)^2}$ and following lemma
\ref{lem:CondDist} in the last inequality.

 \end{proof}
The proof of the lemma to the extreme vertices of $V$ (i.e.
$k=1,2,n-1$ and $n$) is similar and do not bare any difficulty.
However, for $\lambda\ge d-\sqrt{2(d-1)}$ the arguments made above
are insufficient, as in that case $a_1(\lambda)\ge1$. This obstacle
is removed by considering the role of $a_2(\lambda)$ in equation
\ref{eq:ExpK} to show that the event
 \be\nonumber
  \left\{\omega\in\Omega^{+,\alpha}_n,\psi_\omega(v_k)>x\wedge\psi_\omega(v_k)>\frac{1-c_3}{2}(\psi_\omega(v_{k-1})+\psi_\omega(v_{k+1}))\right\}
 \ee
can occur only if for some $v_{k'}\in V$,
$\psi_\omega(v_{k'})-E_{k'}(\omega)$ is proportional to $x$.

As according to the last lemma, the probability to find that
$\psi_\omega(v_k)>x$ is small, unless one of its nearest neighbors
is considerably larger than $x$, theorem \ref{thm:EntropicRepulsion}
follows:
 \begin{proof}\textit{of theorem \ref{thm:EntropicRepulsion}:}
For a given $\lambda$ and $\alpha$, set $\psi_0,c_2,c_3$ as in lemma
\ref{lem:bound_max}.\\
A first observation we make is that if for some
$\omega\in\Omega^{+,\alpha}_n$ and $x>0$
 \be\nonumber
  x\le\psi_\omega(v_{k})\le\frac{1-c_3}2(\psi_\omega(v_{k-1})+\psi_\omega(v_{k+1}))
 \ee
then necessarily
$\max(\psi_\omega(v_{k-1}),\psi_\omega(v_{k+1}))>(1-c_3)^{-1}x$.\\
Now, assume without the loss of generality that
$\psi_\omega(v_{k+1})\ge\psi_\omega(v_{k-1})$. Then, either
$\psi_\omega(v_{k+1})>\frac{1-c_3}2(\psi_\omega(v_{k})+\psi_\omega(v_{k+2}))$,
or $\psi_\omega(v_{k+2})>(1-c_3)^{-2}x$.\\
By iterating the last step (and keeping in mind the convention
$\psi_\omega(v_{-1})\equiv\psi_\omega(v_1),\psi_\omega(v_{n+1})\equiv\psi_\omega(v_{n-1})$),
we find out that if $\psi_\omega(v_k)>x$ then there must exist $1\le
j\le n$, so that $\psi_\omega(v_j)>(1-c_3)^{|j-k|}x$ and in addition
 \be\nonumber
  \psi_\omega(v_j)>\frac{1-c_3}2(\psi_\omega(v_{j-1})+\psi_\omega(v_{j+1}))
 \ee
As a result, according to lemma \ref{lem:bound_max}, we obtain that
 \be\label{eq:GeomSum}
  \phantom{<}\mathbb{P}^{+,\alpha}_n\left(\psi_\omega(v_k)>x\right)\\\nonumber
  <\sum_{j=0}^n \mathbb{P}^{+,\alpha}_n\left(\psi_\omega(v_j)>(1-c_3)^{-|j-k|}x
  \wedge\psi_\omega(v_j)>\frac{1-c_3}2(\psi_\omega(v_{j-1})+\psi_\omega(v_{j+1}))\right)\\\nonumber
  <\frac12\sum_{j=0}^n\exp\left(-c_2(1-c_3)^{-2|j-k|}x^2\right)<c_1 e^{-c_2x^2}
 \ee
where
 \be\nonumber
  c_1=\sum_{j=0}^\infty\exp\left(-c_2[(1-c_3)^{-2j}-1]x^2\right)<\infty
 \ee
 \end{proof}
\noindent In the next section we will be interested in the
conditional distribution $p^n_+(\psi_\omega(v_j))$, where in
addition we condition on $\psi_\omega(v_1),\psi_\omega(v_2)$. For
this purpose we introduce the following variation on theorem
\ref{thm:EntropicRepulsion}:
 \begin{cor}\label{cor:exp_bound_tail}
  $\forall\lambda\in\sigma(T_d)$ and $\alpha\in\mathbb{R}$, $\exists\psi_0(\lambda,\alpha)<\infty$ and $0<c_1(\lambda),c_2(\lambda)<\infty$,
  so that $\forall n\in\mathbb{N}, 3\le k\le n$, $x_1,x_2>\alpha$ and $x>\max(\psi_0,x_2)$
  \be\nonumber
   \mathbb{P}^{+,\alpha}_n\left(\psi_\omega(v_k)\ge x|(\psi_\omega(v_1),\psi_\omega(v_2))=(x_1,x_2)\right)<c_1\e^{-c_2 x^2}
  \ee
 \end{cor}
 \begin{proof}
For a given $\lambda$ and $\alpha$, set $\psi_0,c_2$ and $c_3$ as in
lemma \ref{lem:bound_max}. In a similar manner to the proof of
theorem \ref{thm:EntropicRepulsion}, we notice that if
$\psi_\omega(v_k)\le x$ for some $3\le k\le n$, then there must
exist $3\le j\le n$, so that $\psi_\omega(v_j)>(1-c_3)^{-|j-k|}x$
and in addition
$\psi_\omega(v_j)>\frac{1-c_3}2(\psi_\omega(v_{j-1})+\psi_\omega(v_{j+1}))$.\\
As was shown above, the probability for such an event decays with
$x$ in a Gaussian manner.
 \end{proof}

\section{Phase Transition of the $\alpha$-level sets}\label{sec:PT}
In this section we consider, for a given $\lambda\in\sigma(T_d)$ and
$\alpha\in\mathbb{R}$, the distribution of the large components of
the random process $\{\Omega,\mathbb{P}_\alpha\}$, or the large
$\alpha$-level sets in $\GS$. In particular we prove theorem
\ref{Thm:PT} and the existence of a critical threshold $\alpha_c$,
so that for $\alpha>\alpha_c$ the level sets are almost surely all
finite, while for $\alpha<\alpha_c$ a level-set of an infinite
cardinality will
almost surely appear.\\
Due to the tree structure of $T_d$, we can focus our inquiries in
the following measures over
simple pathes:\\
Let $V=\{v_j\}_{j=1}^n\subset T_d$ be a simple path. We denote
probability densities along the path by the shorthand notation
 \be\nonumber
  p(x_{i_1},x_{i_2}|x_{j_1},x_{j_2})=p((\psi_\omega(v_{i_1}),\psi_\omega(v_{i_2}))=(x_{i_1},x_{i_2})|
  (\psi_\omega(v_{j_1}),\psi_\omega(v_{j_2}))=(x_{j_1},x_{j_2}))
 \ee
and Similarly
 \be\nonumber
  p^{+,\alpha}_n(x_{i_1},x_{i_2}|x_{j_1},x_{j_2})=\\\nonumber
  \quad\quad\quad p((\psi_\omega(v_{i_1}),\psi_\omega(v_{i_2}))=(x_{i_1},x_{i_2})|
  (\psi_\omega(v_{j_1}),\psi_\omega(v_{j_2}))=(x_{j_1},x_{j_2}),v_1\in C_\omega^{\alpha}(v_n))\  .
 \ee
For a given $\lambda\in\sigma(T_d)$ and $\alpha\in\mathbb{R}$ we
define
 \be\nonumber
  P^{(n)}_{\alpha}=\mathbb{P}(v_1\in C_\omega^{\alpha}(v_n))
 \ee
as the probability that a given path of length $n$ is contained in
an $\alpha$- level set. The probability for the same event, where we
condition on $\psi_\omega(v_1),\psi_\omega(v_2)$ will be denoted by
 \be\nonumber
  F^{(n)}_{\alpha}(x_1,x_2)=\mathbb{P}(v_1\in C_\omega^{\alpha}(v_n)|(\psi_\omega(v_1),\psi_\omega(v_2))=(x_1,x_2))
 \ee

The existence of infinite $\alpha$-level sets for small enough
$\alpha$ is proven in \cite{Haggstrom97}, where general invariant
percolation processes on $T_d$ are considered. Using the
mass-transport method it is shown that if the survival probability
of an edge in such a process is larger than $2/d$, an infinite
cluster will appear in almost every realization of the process.
Since for any $\lambda\in\sigma(T_d)$ the survival probability of an
edge in $\GS$ is approaching $1$ as $\alpha\rightarrow-\infty$, the
existence of an infinite component in $T_\alpha(\psi_\omega)$ is
promised $\forall\alpha$ below some (calculable) threshold.\\
The absence of an infinite component for high values of $\alpha$
results from the following lemma:
 \begin{lemma}\label{lem:ExpDec}
  $\forall\lambda\in\sigma(T_d)$ there exist $\beta(\lambda)>0$, so that $\forall\alpha>0$ and $n\in\mathbb{N}$
  \be\nonumber
   P^{(n)}_\alpha<e^{-\beta\alpha^2n}
  \ee
 \end{lemma}
 \begin{proof}
For every $\alpha>0$, $P^{(n)}_\alpha$ can be bounded from above by
 \be\nonumber
  P^{(n)}_\alpha&=&\mathbb{P}(\forall1\le j\le n,\psi_\omega(v_i)>\alpha)\\\nonumber
  &<&\mathbb{P}(\Psi_\omega(V)>n\alpha)
 \ee
where we set $\Psi_\omega(V)=\sum_{j=1}^n\psi_\omega(v_j)$. Note
that $\Psi_\omega$ is a Gaussian random variable, with variance
 \be\nonumber
  \Var(\Psi_\omega(V))&=&\mathbb{E}\left(\sum_{ij}\psi_\omega(v_i)\psi_\omega(v_j)\right)\\\nonumber
  &=&n\left(\phi^{(\lambda)}(0)+2\sum_{j=1}^{n-1}\frac{n-j}n\phi^{(\lambda)}(j)\right)
  <n\Phi^{(\lambda)}
 \ee
where
$\Phi^{(\lambda)}=\phi^{(\lambda)}(0)+2\sum_{j=1}^\infty|\phi^{(\lambda)}(j)|$
(see equation \ref{eq:CofE} and lemma \ref{lem:uniqueness}). Note
that, as $|\phi^{(\lambda)}(j)|=O((d-1)^{-j/2})$, $\Phi^{(\lambda)}$
is finite $\forall\lambda\in\sigma(T_d)$. As a consequence,
 \be\nonumber
  P^{(n)}_\alpha<\frac1{\sqrt{2\pi n\Phi^{(\lambda)}}}\int_{n\alpha}^\infty\exp\left(-\frac{x^2}{2n\Phi^{(\lambda)}}\right)
  <e^{-\beta\alpha^2n}
 \ee
where $\beta=(2\Phi^{(\lambda)})^{-1}$.
 \end{proof}
Recalling that the volume of a sphere in $T_d$ is
$|\Lambda_n|=d(d-1)^{n-1}$, we obtain that
$|\Lambda_n|P^{(n)}_\alpha$ decays exponentially for any
$\alpha>\sqrt{(d-1)/\beta}$, implying that almost surely no infinite
component will appear.\\
 \begin{figure}[h]
  \centering
 \scalebox{0.4}{\includegraphics{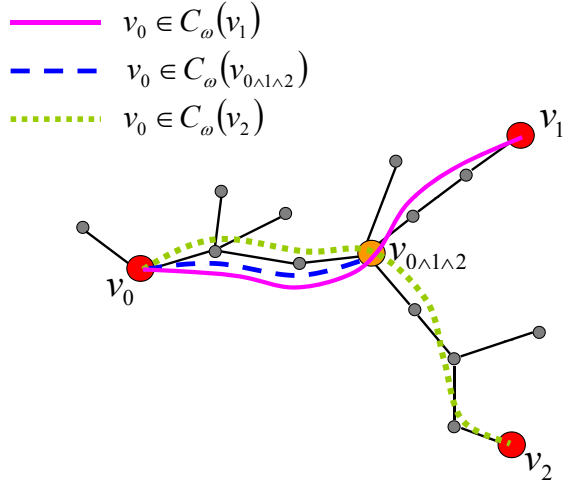}}
  \caption{For a quasi-bernoulli process, the probability to find $v_0\in C_\omega^{\alpha}(v_1)$
   (continuous purple line) conditioned that $v_0\in C_\omega^{\alpha}(v_2)$ (dotted green line)
   is uniformly bounded by the probability that $v_0\in C_\omega^{\alpha}(v_1)$, conditioned that
   $v_0\in C_\omega^{\alpha}(v_0\wedge v_1\wedge v_0)$ (dashed blue line).}
    \label{fig:QBP}
\end{figure}
In order to verify the existence of a critical threshold between the
two phases, we would like to present the following classification of
random processes on trees, introduced in \cite{Lyons}:
 \begin{defn}\label{defn:quasi_bernoulli}
  A random process $\{\Omega,\mathbb{P}\}$ on a tree graph $\Gamma$, associating
  $\forall\omega\in\Omega$ an induced subgraph $\Gamma_\omega\subset\Gamma$,
  is a \textit{quasi Bernoulli} process, if $\exists M<\infty$,
  such that $\forall v_0,v_1,v_2\in\Gamma$:
   \be\label{eq:quasi}
    \frac{\mathbb{P}\left(v_1\in C_{\Gamma_\omega}(v_0)|v_2\in C_{\Gamma_\omega}(v_0)\right)}
    {\mathbb{P}\left(v_1\in C_{\Gamma_\omega}(v_0)|v_{0\wedge1\wedge2}\in C_{\Gamma_\omega}(v_0)\right)}\le M \ .
   \ee
  where $v_{0\wedge1\wedge2}$ is the intersection of the simple paths
  in $\Gamma$ between the three vertices (see figure \ref{fig:QBP})
  and $C_{\Gamma_\omega}(v_0)$ is the connected component of $v_0$ in
  $\Gamma_\omega$.
 \end{defn}
Definition \ref{defn:quasi_bernoulli} provides a simple criterion
for the existence (or the absence) of an infinite component in
$\Gamma_\omega$. For the sake of clarity, we provide here a partial
version of a theorem, derived in \cite{Lyons}:
 \begin{lemma}\label{lem:quasi} \textbf{(Lyons)}
  Let $\{\Omega,\mathbb{P}\}$ be a quasi Bernoulli
  process on $T_d$, which is invariant under the automorphism group of $T_d$
  and associates $\forall\omega\in\Omega$ an induced graph $T_\omega\subset T_d$. If
   \be\nonumber
    \lim_{|v'-v|\rightarrow\infty}\left(\mathbb{P}(v'\in C_{T_\omega}(v))\right)^{1/|v-v'|}<\frac1{d-1}
   \ee
  then, with a high probability, all the connected components of $T_\omega$ are
  finite. If
   \be\nonumber
    \lim_{|v'-v|\rightarrow\infty}\left(\mathbb{P}\left(v'\in C_{T_\omega}(v)\right)\right)^{1/|v-v'|}>\frac1{d-1}
    \ee
  $T_\omega$ will have an infinite component with probability $1$.
  \phantom{move a little bit}$\square$
 \end{lemma}
In order to verify that the level sets of $\GS$ are quasi-Bernoulli,
we provide the following bound on $P^{(n)}_{\alpha}$:
 \begin{lemma}\label{lem:Pnm}
  $\forall\lambda\in\sigma(T_d)$ and $\alpha\in\mathbb{R},\ \exists
  0<c_1,c_2<\infty$ so that $\forall n,m\in\mathbb{N}$
  \be\label{eq:Pnm}
   c_1 P^{(n)}_{\alpha} P^{(m)}_{\alpha}<P^{(n+m)}_{\alpha}<c_2 P^{(n)}_{\alpha} P^{(m)}_{\alpha}
  \ee
 \end{lemma}
Note that the right inequality in equation \ref{eq:Pnm} is the
restriction of equation \ref{eq:quasi} to the case where $v_0$ is
along the simple path between $v_1$ to $v_2$.
 \begin{proof}
According to theorem \ref{thm:EntropicRepulsion} and corollary
\ref{cor:exp_bound_tail}, $\forall \lambda\in\sigma(T_d)$ and
$\alpha\in\mathbb{R}$, there exists a finite threshold
$\alpha<\psi_1(\lambda,\alpha)<\infty$, so that $\forall
n\in\mathbb{N}$ and $j<n$,
 \be\nonumber
  \mathbb{P}^{+,\alpha}_n\left((\psi_\omega(v_j),\psi_\omega(v_{j+1}))\in(\alpha,\psi_1)^2\right)>\frac12 \quad\textrm{ and}\\\nonumber
   \mathbb{P}^{+,\alpha}_{n+2}\left((\psi_\omega(v_3),\psi_\omega(v_4))\in(\alpha,\psi_1)^2|(\psi_\omega(v_1),\psi_\omega(v_2))\in(\alpha,\psi_1)^2\right)>\frac12
 \ee
Recalling that
 \be\nonumber
  \mathbb{P}(v_1\in C_\omega^{\alpha}(v_n)\wedge(\psi_\omega(v_1),\psi_\omega(v_2))\in(\alpha,\psi_1)^2)=
  \int_\alpha^{\psi_1}dy_1dy_2p(y_1,y_2)F^{(n)}_\alpha(y_1,y_2)
 \ee
and that
 \be\nonumber
  \mathbb{P}\left(v_1\in C_\omega^{\alpha}(v_{n+2})\wedge(\psi_\omega(v_3),\psi_\omega(v_4))\in(\alpha,\psi_1)^2
  |(\psi_\omega(v_1),\psi_\omega(v_2))=(x_1,x_1)\right)\\\nonumber
  =\int_\alpha^{\psi_1}dx_3dx_4p(x_3,x_4|x_1,x_2)F^{(n)}_\alpha(x_3,x_4)
 \ee
We obtain by applying bayes' theorem that
$\forall(x_1,x_2)\in(\alpha,\psi_1)^2$,
 \be\label{eq:PnB1}
  \frac12P^{(n)}_\alpha<\int_\alpha^{\psi_1}dy_1dy_2p(y_1,y_2)F^{(n)}_\alpha(y_1,y_2)<P^{(n)}_\alpha
         \\\nonumber
  \frac12 F^{(n+2)}_\alpha(x_1,x_2)<\int_\alpha^{\psi_1}dx_3dx_4p(x_3,x_4|x_1,x_2)F^{(n)}_\alpha(x_3,x_4)<
  F^{(n+2)}_\alpha(x_1,x_2)
 \ee
As $p(y_1,y_2)$ and $p(x_3,x_4|x_1,x_2)$ are strictly positive and
bounded,
$\exists0<\tilde{c}_1(\lambda,\alpha),\tilde{c}_2(\lambda,\alpha)<\infty$
so that for every $\{x_1,x_2,x_3,x_4,y_1,y_2\}\in(\alpha,\psi_1)$:
 \be\nonumber
  \tilde{c}_1<\frac{p(x_3,x_4|x_1,x_2)}{p(y_1,y_2)}<\tilde{c}_2
 \ee
Therefore, we get from equation \ref{eq:PnB1} that
$\forall(x_1,x_2)\in(\alpha,\psi_1)^2$,
 \be\nonumber
  \frac{\tilde{c}_1}{2}P^{(n)}_\alpha<F^{(n+2)}_\alpha(x_1,x_2)<2\tilde{c}_2 P^{(n)}_\alpha
 \ee
Since, by the Markov Property of $\GS$,
 \be\nonumber
  P^{(n+m)}_\alpha=P^{(n)}_\alpha\int_\alpha^\infty dx_{n-1}dx_n
  p_n^+(x_{n-1},x_n)F^{(m+2)}_\alpha(x_{n-1},x_n)
 \ee
We get that
 \be\nonumber
  P^{(n+m)}_\alpha<P^{(n)}_\alpha\int_\alpha^{\psi_1} dx_{n-1}dx_n
  p_n^+(x_{n-1},x_n)F^{(m+2)}_\alpha(x_{n-1},x_n)<c_2P^{(n)}_\alpha
  P^{(m)}_\alpha
 \ee
where $c_2=2\tilde{c}_2\int_\alpha^{\psi_1} dx_{n-1}dx_n
p_n^+(x_{n-1},x_n)$. Similarly,
 \be\nonumber
  P^{(n+m)}_\alpha>\frac12P^{(n)}_\alpha\int_\alpha^{\psi_1} dx_{n-1}dx_n
  p_n^+(x_{n-1},x_n)F^{(m+2)}_\alpha(x_{n-1},x_n)>c_1P^{(n)}_\alpha
  P^{(m)}_\alpha
 \ee
where $c_1=\tilde{c}_1\int_\alpha^{\psi_1} dx_{n-1}dx_n
p_n^+(x_{n-1},x_n)/4$.
 \end{proof}
Note that according to lemma \ref{lem:Pnm},
$\liminf_{n\rightarrow\infty}(P^{(n)}_\alpha)^{1/n}=\limsup_{n\rightarrow\infty}(P^{(n)}_\alpha)^{1/n}$.
As a result, $\lim_{n\rightarrow\infty}(P^{(n)}_\alpha)^{1/n}$
exists and
 \begin{cor}\label{cor:exp_decay}
  $\forall\lambda\in\sigma(T_d)$ and $\alpha\in\mathbb{R}$,
  $\exists 0<c_3,c_4,c_5<\infty$,
  so that $\forall n\in\mathbb{N}$ and $\epsilon>0$,
  \be\nonumber
   c_3 e^{-(c_5+\epsilon)n}\le P^{(n)}_\alpha\le c_4 e^{-(c_5-\epsilon)n}
  \ee
 \end{cor}
The dependence of $F^{(n)}_\alpha(x_1,x_2)$ in its argument, can be
bounded in the following manner
 \begin{lemma}\label{lem:Fn}
  $\forall\lambda\in\sigma(T_d)$ and $\alpha\in\mathbb{R}$, $\exists
  c6,c7<\infty$ so that $\forall n\in\mathbb{N}$ and
  $x_1,x_2>\alpha$
  \be\nonumber
   F^{(n)}_\alpha(x_1,x_2)<c_6\cdot (1+|x_2|)^{c_7}P^{(n)}_\alpha
  \ee
 \end{lemma}
 \begin{proof}
Since $\forall(x_1,x_2)\in\mathbb{R}^2,F^{(n)}_\alpha(x_1,x_2)\le1$,
corollary \ref{cor:exp_decay} implies that $\forall\epsilon>0$,
 \be\nonumber
  F^{(n)}_\alpha(x_1,x_2)<\frac{e^{(c_5+\epsilon)n}}{c_3}P^{(n)}_\alpha
 \ee
Therefore, if $x_2>(d-1)^{n/6}$, the lemma follows
with $c_6=c_3^{-1}$ and $c_4=6c_3/\log(d-1)$.\\
Otherwise, set $X_1=(x_1,x_2)$,
$m=\llcorner6\log_{d-1}(|x_2|)\lrcorner<n$ (where
$\llcorner\cdot\lrcorner$ stands for the integer part). By the
Markov property of $\GS$, $F^{(n)}_\alpha(x_1,x_2)$ equals
 \be\nonumber
  F^{(n)}_\alpha(x_1,x_2)=F^{(m)}_\alpha(x_1,x_2)\int_\alpha^\infty
  dx_{m-1}dx_m p^+_m(x_{m-1},x_m|x_1,x_2) F^{(n-m+2)}_\alpha(x_{m-1},x_m)
 \ee
Since according to corollary \ref{cor:exp_bound_tail},
$\exists\psi_2(\lambda,\alpha)<\infty$ so that
 \be\nonumber
  \mathbb{P}^{+,\alpha}_n\left((\psi_\omega(v_{m-1}),\psi_\omega(v_m))\in(\alpha,x_2+\psi_2)^2|(\psi_\omega(v_1),\psi_\omega(v_2))=(x_1,x_2)\right)>1/2
 \ee
And as
 \be\nonumber
  &&F^{(m)}_\alpha(x_1,x_2) p^+_m(x_{m-1},x_m|x_1,x_2)\\\nonumber
  &=&p(x_{m-1},x_m|x_1,x_2)\mathbb{P}(v_1\in C_\omega^{\alpha}(v_m)|x_1,x_2,x_{m-1},x_m)\\\nonumber
  &<&p(x_{m-1},x_m|x_1,x_2)
 \ee
We obtain that
 \be\label{eq:Fn}\fl
   F^{(n)}_\alpha(x_1,x_2)\le
  2\int_\alpha^{x_1+\psi_2}dx_{m-1}dx_m p(x_{m-1},x_m|x_1,x_2)F^{(n-m+2)}_\alpha(x_{m-1},x_m)
 \ee
Next, we would like to evaluate the conditional density
$p(x_{m-1},x_m|x_1,x_2)$ in terms of $p(x_{m-1},x_m)$. To do so, we
set
 \be\nonumber
  C_{11}=C_{22}=\left(\begin{array}{cc}\phi^{(\lambda)}(0)&\phi^{(\lambda)}(1)\\\phi^{(\lambda)}(1)&\phi^{(\lambda)}(0)\end{array}\right)
   \quad\quad
  C_{12}=\left(\begin{array}{cc}\phi^{(\lambda)}(m-2)&\phi^{(\lambda)}(m-1)\\\phi^{(\lambda)}(m-2)&\phi^{(\lambda)}(m-1)\end{array}\right)
 \ee
and $X_m=(x_{m-1},x_m)$. According to equation \ref{eq:condEC}, the
investigated density is given by
 \be\label{eq:p12m}\fl\quad\quad
  p(x_{m-1},x_m|x_1,x_2)=\frac1{2\pi\sqrt{\det(C)}}\exp\left(-\frac12\left\langle(X_m-\mu),C^{-1}(X_m-\mu)\right\rangle\right)
 \ee
where $\mu=C_{12}C_{22}^{-1}X_0$ and
$C=C_{11}-C_{12}C_{22}^{-1}C_{12}^T$.\\
Since $\phi^{(\lambda)}(m)=O((d-1)^{-m/2})$, we get that
$\|C_{12}\|_2=O((d-1)^{-m/2})$ while $\|C_{11}\|_2=1+|\lambda|/d$.
As a result, by considering the Taylor expansion of $C$ and $\mu$,
we find that
  $C^{-1}=C_{11}^{-1}+C_p$, 
where $\|C_p\|_2=O((d-1)^{-m})$ and
$\|\mu\|_2=O(\|X_1\|_2(d-1)^{-m/2})$.\\
Recalling that $m=\llcorner6\log_{d-1}(1+\|X_1\|_\infty)\lrcorner$,
we obtain from equation \ref{eq:p12m} that $\forall
x_{m-1},x_m<x_2+\psi_2$,
 \be\nonumber
  p(x_{m-1},x_m|x_1,x_2)=\frac1{2\pi\sqrt{\det(C)}}\exp\left(-\frac12\left\langle X_m,C_{11}^{-1}X_m\right\rangle+f(X_1)\right)
 \ee
where
 \be\nonumber
  f(X_1)=-\frac12\langle X_m,C_p X_m\rangle-\langle X_m, C^{-1}\mu\rangle
  \underset{\|X_1\|\rightarrow\infty}\rightarrow0
 \ee
is a bounded function of $X_1$. Since
 \be\nonumber
  p(x_{m-1},x_m)=\frac1{2\pi\sqrt{\det(C_{11})}}\exp\left(-\frac12\langle X_m,C_{11}^{-1}X_m\rangle\right)
 \ee
and as $\det(C_{11})/\det(C)=1+O((d-1)^{-m})<2$, we find that
 \be\nonumber
  p(x_{m-1},x_m|x_1,x_2)<\tilde c_1 p(x_{m-1},x_m)
 \ee
where $\tilde c_1(\lambda,\alpha)=2\max_{\mathbb{R}^2}
f(X_1)<\infty$. Returning to equation \ref{eq:Fn}, we find that
 \be\nonumber
   F^{(n)}_\alpha(x_1,x_2)&\le&
  \tilde c_1\int_\alpha^{x_1+\psi_2}dx_{m-1}dx_m p(x_{m-1},x_m)F^{(n-m+2)}_\alpha(x_{m-1},x_m)\\\nonumber
  &<&\tilde c_1\int_\alpha^\infty dx_{m-1}dx_m p(x_{m-1},x_m)F^{(n-m+2)}_\alpha(x_{m-1},x_m)\\\nonumber
  &=&\tilde c_1 P^{(n-m+2)}_\alpha
 \ee
Finally, as according to corollary \ref{cor:exp_decay},
$\forall\epsilon>0$,
$P^{(n-m+2)}_\alpha\le(c_4/c_3)e^{(c_5+\epsilon)\cdot(m-2)}P^{(n)}_\alpha$,
 \be\nonumber\fl \textrm{setting }
  c_6=\frac{\tilde c_1 c_4}{c_3},\
  c_7=\frac{6c_5}{\log(d-1)},\textrm{ the lemma follows.}
 \ee
 \end{proof}
Having lemmas \ref{lem:Pnm} and \ref{lem:Fn} in hand, we are ready
to prove theorem \ref{Thm:PT}:
 \begin{proof}\textit{of theorem \ref{Thm:PT}}:
First, we note that $\forall\lambda\in\sigma(T_d)$ and
$\alpha\in\,\mathbb{R}$, the $\alpha$-level sets of $\GS$ are
quasi-Bernoulli.\\
Indeed, let $v_0,v_1,v_2\in T_d$. If $v_1$ ($v_2$) is on the pass,
connecting $v_0$ to $v_2$ ($v_1$), condition \ref{eq:quasi} is
trivially satisfied, with $M=1$. if $v_0$ connects $v_1$ to $v_2$,
then according to lemma \ref{lem:Pnm}, condition \ref{eq:quasi} is
satisfied with $M=c_2(\lambda,\alpha)$.\\
Otherwise, $v_{0\wedge1\wedge2}\notin\{v_0,v_1,v_2\}$. We denote by
$v_0',v_1',v_2'$ the vertices which are adjacent to
$v_{0\wedge1\wedge2}$ on the simple path leading to $v_0,v_1,v_2$
correspondingly; We rewrite the LHS of equation \ref{eq:quasi} as:
  \be\nonumber\fl\quad\quad\quad
   \frac{\mathbb{P}\left(v_1\in C_\omega^{\alpha}(v_0)|v_2\in C_\omega^{\alpha}(v_0)\right)}
   {\mathbb{P}\left(v_1\in C_\omega^{\alpha}(v_0)|v_{0\wedge1\wedge2}\in C_\omega^{\alpha}(v_0)\right)}=
    \frac{\mathbb{P}\left(\{v_1,v_2\}\subset C_\omega^{\alpha}(v_0)\right)\mathbb{P}\left(v_{0\wedge1\wedge2}\in C_\omega^{\alpha}(v_0)\right)}
    {\mathbb{P}\left(v_1\in C_\omega^{\alpha}(v_0)\right)\mathbb{P}\left(v_2\in C_\omega^{\alpha}(v_0)\right)}
  \ee
and set $n_j=|v_{0\wedge1\wedge2}-v_j|$ for $j=0,1,2$.\\
By the Markov property of $\GS$, we express:
 \be\nonumber
  \mathbb{P}\left(\{v_1,v_2\}\subset C_\omega^{\alpha}(v_0)\right)&=&\int_{-\infty}^\alpha dx_{0\wedge1\wedge2}dx_0'dx_1'dx_2'\\\nonumber
  &\times& p(x_{0\wedge1\wedge2},x_0',x_1',x_2')\cdot\prod_{j=0,1,2}F^{(n_j)}_\alpha(x_{0\wedge1\wedge2},x_j')
 \ee
Following lemma \ref{lem:Fn}, the RHS is bounded by
 \be\nonumber
  \mathbb{P}\left(\{v_1,v_2\}\subset C_\omega^{\alpha}(v_0)\right)\le
  c_8(\lambda,\alpha)\prod_{j=0,1,2}P^{(n_j)}_\alpha
 \ee
where
 \be\nonumber
  c_8(\lambda,\alpha)=c_6^3\int_{-\infty}^\alpha dx_{0\wedge1\wedge2}dx_0'dx_1'dx_2'
   p(x_{0\wedge1\wedge2},x_0',x_1',x_2')\prod_{j=0,1,2}(1+|x_j|)^{c_7}
 \ee
Note that as $c_8(\lambda,\alpha)$ is a moment of a Gaussian
distribution, it is finite.\\
Finally, following lemma \ref{lem:Pnm}, we get that
 \be\nonumber
  \frac{\mathbb{P}\left(\{v_1,v_2\}\subset C_\omega^{\alpha}(v_0)\right)\mathbb{P}\left(v_{0\wedge1\wedge2}\in C_\omega^{\alpha}(v_0)\right)}
  {\mathbb{P}\left(v_1\in C_\omega^{\alpha}(v_0)\right)\mathbb{P}\left(v_2\in
  C_\omega^{\alpha}(v_0)\right)}<\frac{c_8(\lambda,\alpha)}{c_1^2(\lambda,\alpha)}
 \ee
establishing by that the quasi-Bernoulli property of
$\{\Omega,\mathbb{P}_\alpha\}$.\\
As (following lemma \ref{lem:Pnm}) for $\alpha>\sqrt{(d-1)/\beta}$,
$\lim_{n\rightarrow\infty}(P^{(n)}_\alpha)^{1/n}<1/(d-1)$, while
according to \cite{Haggstrom97},
$\lim_{n\rightarrow\infty}(P^{(n)}_\alpha)^{1/n}>1/(d-1)$ for small
enough $\alpha$ and since
$\lim_{n\rightarrow\infty}(P^{(n)}_\alpha)^{1/n}$ is strictly
decreasing in $\alpha$, theorem \ref{Thm:PT} follows, where
$\alpha_c$ is given by the (implicit) expression
 \be\nonumber
  \lim_{n\rightarrow\infty}(P^{(n)}_{\alpha_c})^{1/n}=\frac1{d-1}
 \ee
 \end{proof}

\appendix
\section{proof to lemma \ref{lem:bound_max}}\label{app:bound_max}
In section \ref{sec:entropic_repulsion} we have established lemma
\ref{lem:bound_max} for vertices in the bulk of the path ($2<k<n-1$)
and $\lambda<d-\sqrt{2(d-1)}$. We begin by proving the lemma
$\forall\lambda\in\sigma(T_d)$ for the case $k=1$ ($k=n$). The main
theme in the proof is the partition of the event
 \be\nonumber
  \left\{\omega\in\Omega^{+,\alpha}_n,\psi_\omega(v_k)>x\wedge\psi_\omega(v_k)>\frac{1-c_3}2(\psi_\omega(v_{k-1})+\psi_\omega(v_{k+1}))\right\}
 \ee
into a finite union of events, so that in every subevent,
$\psi_\omega(v_{k'})$ exceed $E_{k'}(\omega)$ significantly (see
equation \ref{eq:ExpK}) for some $k'$.
\subsubsection*{the case $k=1$}:\\
Given $\psi_0$ and $c_3$, we would like to evaluate the probability
of the event
 \be\nonumber
  A(x)&=&\left\{\omega\in\Omega^{+,\alpha}_n,\psi_\omega(v_1)>x\wedge\psi_\omega(v_1)>(1-c_3)\psi_\omega(v_2)\right\}
 \ee
for $x>\psi_0$ by decomposing it into
\be\nonumber
  A_1(x)=\left\{\omega\in A,\psi_\omega(v_3)>\lambda\psi_\omega(v_2)-\frac{d-1}{1+c_3}\psi_\omega(v_1)\right\}\\\nonumber
  A_2(x)=\left\{\omega\in A\setminus A_1,\psi_\omega(v_4)>\lambda\psi_\omega(v_1)
                +\frac{d\lambda}{d-1}\psi_\omega(v_3)-\frac{\lambda^2+(d-1)^2}{(d-1)(1+c_3)}\psi_\omega(v_2)\right\}\\\nonumber
  A_3(x)=\left\{\omega\in A\setminus \bigcup_{j=1,2}A_j,\psi_\omega(v_5)>\right.\\\nonumber
                \phantom{bgfnbtntynjkghkljg}\left.\frac{\lambda d}{d-1}(\psi_\omega(v_2)+\psi_\omega(v_4))-\psi_\omega(v_1)
                -\frac{\lambda^2+(d-1)^2+1}{(d-1)(1+c_3)}\psi_\omega(v_3)\right\}\\\nonumber
  A_4(x)=A\setminus\bigcup_{j=1}^3A_j
 \ee
This partition is chosen, following equation \ref{eq:ExpK}, so that
$\forall1\le j\le3$ and $\omega\in A_j(x)$,
 \be\nonumber
  \psi_\omega(v_j)-E_j(\omega)>\beta_j x
 \ee
where $\beta_j$ are some strictly positive functions of $c_3$.
Therefore, according to lemma \ref{lem:CondDist} and equation
\ref{eq:CondExp}, we find out that $\exists\psi_0(\lambda,\alpha)$,
so that $\forall x>\psi_0$,
 \be\nonumber
  \mathbb{P}^{+,\alpha}_n(\omega\in A_j(x))<\exp\left(-\frac{\beta_j^2x^2}{2}\right)
 \ee
As a result, the lemma will follow by showing that
$\forall\lambda\in\sigma(T_d)$ and $\alpha\in\mathbb{R}$, there
exists $c_3(\lambda,\alpha)>0$ and $\psi_0(\lambda,\alpha)<\infty$,
so that $\forall x>\psi_0$, $A(x)=\bigcup_{j=1}^3A_j(x)$.\\
For $\lambda\le0$, as $\forall\omega\in A$,
$\psi_\omega(v_2)>\alpha$, we obtain that
 \be\nonumber
  A_1(x)\supset\left\{\omega\in A,\alpha>-|\lambda|\psi_\omega(v_2)-\frac{d-1}{1+c_3}\psi_\omega(v_1)\right\}
 \ee
Therefore $\forall c_3>0$, setting
$\psi_0=-(1+|\lambda|)(1+c_3)\alpha/(d-1)$, we find that
$A(x)\subset A_1(x)$ for every $x>\psi_0$.\\
Similarly, if $0<\lambda\le d-1$ then
 \be\nonumber
  A_1(x)\supset\left\{\omega\in A,\alpha>\frac{\lambda-(d-1)}{1+c_3}\psi_\omega(v_1)\right\}
 \ee
Therefore $\forall c_3>0$, setting
$\psi_0=(1+c_3)\alpha/(d-1-\lambda)$, we get that $A(x)\subset
A_1(x)$ for every $x>\psi_0$. Note that if $d>5$ then
$\max(\sigma(T_d))=2\sqrt{d-1}<d-1$, and the proof is done.\\
If $d-1\le\lambda<(d-1+\sqrt{d^2+2d-3})/2$, we find that
 \be\nonumber
  A_2(x)\supset{\textstyle\left\{\omega\in A\setminus A_1,\alpha>\left(-\left(1+\frac{dc_3}{d-1}\right)\lambda^2
  +(d-1-(d+c_3)c_3)\lambda+d-1\right)\frac{\psi_\omega(v_1)}{1-c_3}\right\}}
 \ee
As in the limit $\psi_\omega(v_1)\rightarrow\infty, c_3\rightarrow0$
and $\forall d-1<\lambda<(d-1+\sqrt{d^2+2d-3})/2$
 \be\nonumber
  \alpha>\left(-\left(1+\frac{dc_3}{d-1}\right)\lambda^2
  +(d-1-(d+c_3)c_3)\lambda+d-1\right)\frac{\psi_\omega(v_1)}{1-c_3}
 \ee
we obtain that for any $\lambda$ smaller then
$(d-1+\sqrt{d^2+2d-3})/2$, there exist finite and positive
$\psi_0,c_3$, so that $\forall x>\psi_0$, $A(x)=A_1(x)\cup
A_2(x)$.\\
As $\max(\sigma(T_d))<(d-1+\sqrt{d^2+2d-3})/2$ for $d\ge3$, we are
left with the case $d=3$ and $1+\sqrt3<\lambda\le2\sqrt2$.
Preforming a similar calculation, one finds that for an appropriate
choice of $\psi_0$ and $c_3$, $A_3(x)\supset
A(x)\setminus\bigcup_{j=1,2}A_j(x)$, as long as
$\lambda^3-2\lambda^2-4\lambda+4<0$. As this is indeed the case for
every $1+\sqrt3<\lambda\le2\sqrt2$, the proof is complete.
\subsubsection*{the case $k>2$}:\\
The proof is similar to the above, but require few more iterations.
Using the shorthand notation
 \be\nonumber
  \psi_\omega^{(j)}(v_k)=\frac12\left(\psi_\omega(v_{k-j})+\psi_\omega(v_{k+j})\right)
 \ee
we decompose, for a given $\psi_0$ and $c_3$, the event
 \be\nonumber
  A(x)&=&\left\{\omega\in\Omega^{+,\alpha}_n,\psi_\omega(v_k)>x\wedge\psi_\omega(v_k)>(1-c_3)\psi_\omega^{(1)}(v_k)\right\}
 \ee
into
 \be\nonumber
  A_1(x)=\left\{\omega\in A,\psi_\omega^{(2)}(v_k)>\frac{d\lambda}{d-1}\psi_\omega^{(1)}(v_k)-
  \frac{\lambda^2+(d-1)^2+1}{(d-1)(1+c_3}\psi_\omega(v_k)\right\}\\\nonumber
  A_2(x)=\left\{\omega\in A\setminus A_1,\psi_\omega^{(3)}(v_k)>\right.\\\nonumber
            \phantom{bgfnbjg}\left.\frac{d\lambda}{d-1}(\psi_\omega(v_k)+\psi_\omega^{(2)}(v_k))
            -\left(\frac{\lambda^2+(d-1)^2+1}{(d-1)(1+c_3)}-1\right)\psi_\omega^{(1)}(v_k)\right\}\\\nonumber
  A_j(x)=\left\{\omega\in A\setminus\bigcup_{i=1}^{j-1}A_i,\psi_\omega^{(j+1)}(v_k)>\right.\\\nonumber
            \phantom{bgfnbjg}\left.\frac{d\lambda}{d-1}(\psi_\omega{(j-2)}(v_k)+\psi_\omega^{(j)}(v_k))
            -\psi_\omega^{(j-3)}(v_k)-\frac{\lambda^2+(d-1)^2+1}{(d-1)(1+c_3)}\psi_\omega^{(j-1)}(v_k)\right\}
 \ee
where we assume that $x>\psi_0$ and $j\ge3$.\\
As before, the intervals are chosen so that if $\omega\in A_{j+1}$
for some $j\ge0$, then
 \be\nonumber
  \psi_\omega^{(j)}(v_k)-\frac12(E_{k-j}(\omega)+E_{k+j}(\omega))>\beta_j x
 \ee
for some $\beta_j(c_3)>0$, implying that for large enough $\psi_0$
 \be\nonumber
  \forall x>\psi_0,\ \mathbb{P}^{+,\alpha}_n(\omega\in A_j(x))<\exp\left(-\frac{\beta_j^2x^2}{2}\right)
 \ee
As a result, the lemma will follow by verifying that
$\forall\lambda\in\sigma(T_d)$ and $\alpha\in\mathbb{R}$ there exist
$\psi_0<\infty$ and $c_3>0$ so that $\forall x>\psi_0,
A(x)=\bigcup_{j=1}^{p(\lambda)}A_j(x)$ for some finite integer
$p(\lambda)$.\\
As was demonstrated in the (partial) proof at section
\ref{sec:entropic_repulsion}, if $\lambda<d-\sqrt{2(d-1)}$ then for
an appropriate choice of $\psi_0$ and $c_3$, we find that $A(x)=A_1(x)$.\\
If $\lambda\ge d-\sqrt{2(d-1)}$, we find that $A_2(x)=A\setminus
A_1(x)$ for small enough $c_3$ and large enough $\psi_0$, as long as
 \be\label{eq:L3}\fl\quad\quad
  -d\lambda^3+2(d^2-d+1)\lambda^2-d(d^4-4d+4)\lambda-2d^3+4d^2-4d+2<0
 \ee
The last polynomial has a single real root
$\lambda_0$\footnote{$\lambda_0=\frac{2(d^2-d+1)}{3d}+\frac{(d^2+2d-2)^2}
{3d(d^6+6d^5-21d^4+38d^3-39d^2+24d-8+3^{1.5}d(d-1)(-2d^6-12d^5+15d^4-22d^3+51d^2-48d+16)^{1/2})^{1/3}}
+\frac{(d^6+6d^5-21d^4+38d^3-39d^2+24d-8+3^{1.5}d(d-1)(-2d^6-12d^5+15d^4-22d^3+51d^2-48d+16)^{1/2})^{1/3}}{3d}$},
where $\forall d>5,\lambda_0>2\sqrt{d-1}$. Therefore, for these
cases, condition \ref{eq:L3} is fulfilled and the lemma follows.\\
Iterating the process four more times (where each iteration involves
the evaluation of the roots of a polynomial of increasing degree),
we find for $d=4,5$ that $\forall\lambda\in\sigma(T_d)$,
$A(x)=\bigcup_{j=1}^3A_j(x)$ (for an appropriate choice of $\psi_0$
and $c_3$). For $d=3$ we get that $A(x)=\bigcup_{j=1}^6A_j(x)$, by
that establishing the lemma for $2<k<n-1$.\\
The proof for the case $k=2$ ($k=n-1$) is identical and therefore
will be omitted.

\ack{I would like to thank U. Smilansky, I. Benjamini, M. Aizenman
and O. Zeitouni for enlightening discussions, comments and suggestions.\\
The work was supported by the Minerva Center for non-linear Physics,
the Einstein (Minerva) Center at the Weizmann Institute, and by
grants from the ISF (grant 166/09), GIF (grant 284/3), BSF
(710021/1) and Afeka college of engineering.

\newpage
\noindent {\bf{Bibliography}}
\bibliographystyle{unsrt}
\bibliography{geometric}
\end{document}